\documentclass[runningheads]{llncs}

\usepackage{amssymb,amsmath,multicol,enumerate,amsthm}

\usepackage{blindtext}
\usepackage[utf8]{inputenc}

\usepackage{latexsym}
\usepackage{algorithm}
\usepackage{algorithmic}
\usepackage{graphicx}
\usepackage{blindtext}
\usepackage{color}
\usepackage{tikz}
\begin{document}
	
\title{Unit Perturbations in Budgeted Spanning Tree Problems
}

\author{Hassene Aissi\inst{1}           \and 
        Solal Attias\inst{2} \and
        Da Qi Chen\inst{3} \and
        R. Ravi\inst{4}\thanks{This material is based upon work supported in part by the U. S. Office of Naval Research under award number N00014-18-1-2099 and the Air Force Office of Scientific Research under award number FA9550-20-1-0080}
}

\institute{Paris Dauphine University \\
              %Tel.: +123-45-678910\\
              %Fax: +123-45-678910\\
              \email{aissi@lamsade.dauphine.edu}           %  \\
%             \emph{Present address:} of F. Author  %  if needed
           \and
           {Ecole Normale Superieure\\
               \email{solal.attias@ens.fr}}
           \and
           {Biocomplexity Institute and Initiatives at University of Virginia \\
              \email{wny7gj@virginia.eduu}} \\
           \and
           {Tepper School of Business, Carnegie Mellon University \\
           	\email{ravi@cmu.edu}} \\
}

%\input{cover}
%\newpage

\maketitle

\abstract{The minimum spanning tree of a graph is a well-studied structure that is the basis of countless graph theoretic and optimization problem. We study the minimum spanning tree (MST) perturbation problem where the goal is to spend a fixed budget to increase the weight of edges in order to increase the weight of the MST as much as possible. Two popular models of perturbation are bulk and continuous. In the bulk model, the weight of any edge can be increased exactly once to some predetermined weight. In the continuous model, one can pay a fractional amount of cost to increase the weight of any edge by a proportional amount.  Frederickson and Solis-Oba \cite{FS96} have studied these two models and showed that bulk perturbation for MST is as hard as the $k$-cut problem while the continuous perturbation model is solvable in poly-time.

In this paper, we study an intermediate unit perturbation variation of this problem where the weight of each edge can be increased many times but at an integral unit amount every time. We provide an $(opt/2 -1)$-approximation in polynomial time where $opt$ is the optimal increase in the weight. We also study the associated dual targeted version of the problem where the goal is to increase the weight of the MST  by a target amount while minimizing the cost of perturbation. We provide a $2$-approximation for this variation. Furthermore we show that assuming the Small Set Expansion Hypothesis, both problems are hard to approximate.

We also point out an error in the proof provided by Frederickson and Solis-Oba in \cite{FS96} with regard to their solution to the continuous perturbation model. Although their algorithm is correct, their analysis  is flawed. We provide a correct proof here.
}

\clearpage

\section{Introduction}
Classic problems in network optimization involve minimizing (or maximizing) the cost of interesting combinatorial objects such as cuts, paths or trees. 
The motivation for Ford and Fulkerson's study of the minimum $st$-cut problem~\cite{harris55,schrijver2002history,chestnut2017hardness} 
was to examine how much an enemy can interdict the network with a limited budget and hence reduce the capacity of the minimum cut to the lowest possible value under different costs for capacity reduction at the arcs. 
The minimum $st$-cut problem is posed in a directed graph with non-negative arc capacities and special source $s$ and sink $t$ with the goal of finding the minimum capacity of a cut separating all paths from $s$ to $t$.
In early work, Fulkerson~\cite{FulkersonMS} also introduced a budgeted version of this minimum $st$-cut problem. 
Here in addition to the capacities, each arc is also provided with a cost to increase its capacity by one unit and a global budget $B > 0$ on the total cost that can be spent with the goal of maximizing the minimum cut after carrying out the increases within budget. 
Ahuja and Orlin~\cite{AO95} later showed how this algorithm can be made to run in strongly polynomial time.
%Phillips~\cite{[hillips93} studied a version of this problem wherein each arc is given a destruction cost and by spending an $0< \alpha \leq 1$ fraction of the destruction cost from the budget, an $\alpha$-fraction of the capacity of the corresponding arc is removed (rather than increased): this can be viewed as a version of the opposite problem of attacking the network. 

Fulkerson and Harding~\cite{FH77} studied the analogue of this problem for the case of minimum-length $st$-paths.
In the shortest $st$-path problem in a directed graph with non-negative arc lengths and special source $s$ and destination $t$, the goal is to find the minimum length path from $s$ to $t$.
The budgeted version addressed in~\cite{FH77} specifies an integral cost per unit increase of the length for each arc and a budget $b$ and seeks to maximize the shortest $st$-path under this constraint. They reduce this problem to an instance of parametric minimum-cost flow which is also known to be solvable in polynomial time~\cite{tardos85}.

The study of budgeted problems for the case of minimum weight spanning trees in undirected graphs was initiated by Frederickson and Solis-Oba~\cite{FS99}.
The minimum spanning tree problem in a given undirected graph with non-negative edge weights is to find a spanning tree of minimum total edge weight. 
The budgeted version specifies costs per unit increase in the weight of the edges and a budget $B>0$ and requires spending the budget optimally to increase the edge-weights to maximize the weight of the resulting minimum-weight spanning tree.
Frederickson and Solis-Oba gave a strongly polynomial algorithm for this problem. 
In follow-up work~\cite{FS98}, they extended this solution to that of maximizing the minimum weight base of matroids for which independence can be tested in polynomial time.

Juttner~\cite{juttner06} extended the study of strongly polynomial time algorithms for such budgeted optimization problems to a larger class extending the work of~\cite{FS98} to the budgeted cases of minimum-cost circulations,  matroid intersections and submodular optimization.

\subsection{Budgeted versus Target Versions}
In contrast to requiring that the total cost-spent is no more than a pre-specified budget $B$ so as to maximize the minimum-weight solution, an alternate formulation is to specify a target final value $T$ for the minimum weight solution and minimize the total cost spent in increasing the weights of the edges to achieve this target. Note that a polynomial time algorithm for one version usually yields a polynomial time algorithm for the other, by carrying out a binary search over the specified parameter and optimizing the other; however, this strategy does not directly yield strongly polynomial algorithms. We will study both the original budgeted and the dual targeted versions of the problem in this paper.

\subsection{Discrete Perturbation Models}
We remark that all the above budgeted optimization formulations allow fractional changes, i.e., if we increase the weight of an edge $e$ by any $\delta>0$, we pay $\delta \cdot c(e)$ where $c(e)$ is the cost per unit weight increase for the edge $e$. We may think of this as the {\em continuous} version of these problems. 
A natural direction to investigate is the {\em discrete} version of these problems where the changes to the elements (edges) are only allowed to be discrete units. 

Frederickson and Solis-Oba~\cite{FS99} introduce an even simpler version of the discrete version of the robust minimum spanning tree problem: for a given value $k$, find the set of $k$ edges whose removal from the graph causes the largest increase in the weight of its minimum spanning trees.
They point out that this problem already appears in the literature under the name of the most vital edges in the MST~\cite{HJLH91,HWW92,IK93,LC93,tarjan82}.
The most vital arcs problem has also been studied in the context of shortest $st$-paths~\cite{BGV89,MMG89}.

\subsection{Bulk versus Unit Perturbations}
We distinguish between two types of discrete perturbation models: bulk and unit. In the \textbf{bulk perturbation} models, each edge is specified with a bulk reduction in its weight for which a modification cost must be paid in full for any modification to be effected at all. More generally, each edge has a series of integer-valued successive modifications in its weight that can be effected by a corresponding set of modification costs, each of which must be paid in full to effect these modifications.
In contrast, the \textbf{unit perturbation} model specifies for each edge, a unit modification cost, and bounds (upper or lower, depending on whether the modification is an increase or decrease of weight respectively) and modifications can be made one unit at a time until the bound. 
In both models, in the budgeted version, we are given a budget $B$ on the total modification budget and the objective as before is to increase the cost of a minimum cost solution. Similarly, for the targeted version, we are given a target $T$ on the final value of the objective and the goal is to minimize the modification cost so as to reach the target value $T$ in the objective after these modifications. 

As a simple example, consider the targeted unit-perturbation version of the $st$-minimum cut problem in a digraph with edge capacities. Assume that we have an estimate of fortification cost $r_a$ for every arc to increase its capacity by one unit; we assume that the fortification costs are valid for any starting capacity in a given range. 
Given a target $T$, the goal is to increase the capacities on a subset of arcs at minimum fortification cost so that the minimum $st$-cut in the fortified network is at least $T$. 
We may have individual limits on how much each arc can be fortified in total. 
We insist that the total capacity increase in each arc is integral to retains the combinatorial flavor in the discrete problem. This model is ``smoother" than the bulk model where every edge has a fixed integral capacity increase (potentially greater than one) for which the full upgrading cost has to be paid: any smaller use of the budget does not upgrade the edge at all.
In contrast to the more popular interdiction problems modeling an attacker's concerted effort to weaken a network, fortification (such as infrastructure replenishment) is typically carried out over time.
The unit-upgrading is more natural in this context of incremental upgrading of elements over time.
Indeed, it is known that under the bulk model that requires us to pay a single cost to increase the capacity by an arbitrarily large amount,  the minimum $st$-cut fortification problem is as hard as the notorious densest subgraph (DkS) problem (See~\cite{chestnut2017hardness} for a reduction).
Henceforth we focus our attention in the rest of this paper on the unit perturbation models.

\subsection{Problem Formulations}
In this paper, we study both the targeted and budget version of the minimum weight spanning tree problem under unit perturbations. We use $(G, w)$ to denote the input graph $G$ with edge weights $w$. Given an upgrading scheme, let $c_x := c \cdot x = \sum_{e\in E(G)} c_ex_e$ is the total cost of this upgrading. 
\begin{problem} (Budgeted MMST)
	Given a graph $G$, functions $w, c:E(G)\to \mathbb{N}$ and a budget $B$, let $f_x$ be the weight of an MST in $(G,w+x)$. Then, find an upgrading scheme $x:E(G)\to \mathbb{N}$ that maximizes $f_x-f_0$ while ensuring the total cost, $c_x := c \cdot x$, is at most $B$. We call $f_G(B)$ this maximum value. 
\end{problem}

Note that we have modelled the increase in the weight of the MST as the target, which is more general/powerful than making the target the final weight of the MST. 

We also consider the targeted version of this problem:
\begin{problem} (Targeted MMST)
	Given a graph $G$, functions $w, c:E(G)\to \mathbb{N}$ and a target increase of $T-T_0$ where $T_0$ is the weight of the MST in $(G, w)$, find a perturbation scheme $x:E(G)\to \mathbb{N}$ that minimizes the cost $c_x := c \cdot x$ while ensuring that $f_x$ the weight of any MST in $(G,w+x)$, is at least $T$.
\end{problem}

We define two related versions of the problems we study.

\begin{problem} (Continuous Budgeted MMST)
	Given a graph $G$, functions $w, c:E(G)\to \mathbb{N}$ and a budget $B$, let let $f_x$ be the weight of an MST in $(G,w+x)$. . Then, find a upgrading scheme $x:E(G)\to \mathbb{R^+}$ that maximizes $f_x-f_0$ while ensuring the total cost, $c_x := c \cdot x$, is at most $B$. We call $f_G(B)$ this maximum value.
\end{problem}

\begin{problem} (Continuous Targeted MMST)
	Given a graph $G$, functions $w, c:E(G)\to \mathbb{N}$ and a target increase value of $T-T_0$ where $T_0$ is the weight of the MST in $(G, w)$, find a upgrading scheme $x:E(G)\to \mathbb{R^+}$ that minimizes the cost $c_x := c \cdot x$ while ensuring that $f_x$ the weight of any MST in $(G,w+x)$, is at least $T$.
\end{problem}

Frederickson and Solis-Oba study the above continuous version of the Budgeted MMST problem and provide an optimal polynomial time algorithm \cite{FS96}, where the key difference is that the upgrading amounts on the edges are allowed to be fractional rather than integral values. 
%A continuous Targeted version of MMST can be analogously defined.

All of the above problems permit a natural generalization with individual upper bounds on the fortification or upgrading amounts on the edges.
\begin{problem} (MMSTU)
Given a graph $G$ and functions $w, c, u :E(G)\to \mathbb{N}$ and a budget $B$, find an upgrading scheme $x:E(G)\to \mathbb{N}$, with $\forall e\in E(G), x(e) \leq u(e)$, that maximizes $f_x-f_0$ while ensuring $c_x\le B$.
\end{problem}
We notice that the algorithm of Frederickson and Solis-Oba~\cite{FS96} is optimal even with upper bounds, since the proof is not changed with upper bounds.

\subsection{Contributions}

\begin{enumerate}
    \item Frederickson and Solis-Oba~\cite{FS96} prove that for the budgeted version of Continuous MMST an optimal solution can be found in polynomial time. However, there is an error in the proof of the optimality of their algorithm. We point out the mistake and provide a correct proof. (Section~\ref{sec:correct})
    \item We provide a 2-approximation for the targeted version of the unit-perturbation MMST. For the budgeted version, we provide a $\frac{opt}{2}-1$-solution, an algorithm that gives an upgrading scheme where the increase in the weight of the MST is at least $\frac{opt}{2}-1$, where $opt$ is the optimal increase in the weight of the MST. (Section~\ref{sec:approx})
    \item 
    %We show that the Targeted MMST is NP-hard with upper bounds on the weight increases via a reduction from minimum-weight $k$-cut problem where the goal is to find a minimum weight set of edges to break the graph into at least $k$ connected components.
    %We then extend this to the case with all edges having unit upgrading costs, no upper bounds but with parallel edges. Similarly, we show that Budgeted MMST is NP-Hard by a reduction from the complementary Max Components problem of using a given budget on the weights to delete edges (of total weight within this budget) so as to maximize the number of resulting connected components.
    We show that both versions of MMST are NP-complete even with unit upgrading costs. We use reductions from the minimum-weight $k$-cut problem where the goal is to find a minimum weight set of edges to break the graph into at least $k$ connected components, and from the complementary Max Components problem of using a given budget to delete edges (of total weight within this budget) so as to maximize the number of resulting connected components.
    By using appropriate approximation-preserving reductions and the known hardness results for these two problems, we get the following implications that show that our approximation algorithms are (nearly) best possible: Assuming the Small Set Expansion Hypothesis (SSEH)~\cite{Manu18}, Targeted Discrete MMST is NP-hard to approximate to a $(2 - \varepsilon)$-factor for $\varepsilon > 0$. Similarly, assuming SSEH, Budgeted Discrete MMST is NP-hard to approximate to a $(\frac{1}{2} + \varepsilon)$-factor for $\varepsilon > 0$. (Section~\ref{sec:hardness})
    \item The complexity of MMST with unit costs for upgrading any edge and the same starting weights for all edges in the graph is unresolved. Hence we study this special case of the Targeted MMST on an undirected graph with $n$ nodes where all edges have the same initial starting weight and MST weight target $T$, and give an optimal algorithm that runs in $n^{O(T)}$ time. For this, we use a supermodularity property of a coverage function used in the analysis. For the analogous budgeted version with budget $B$, this also implies an optimal algorithm running in time $n^{O(B)}$. (Section~\ref{sec:samestart})
\end{enumerate}

\subsection{Related Work}

Somewhat tangentially related to our work is the attacker's problem of interdicting the network to {\em decrease} the
weight of the minimum weight structure in the network.
Such interdiction problems have a rich and long history in Combinatorial Optimization including serving as the motivation for Ford and Fulkerson's study of the minimum $st$-cut problem~\cite{harris55,schrijver2002history,chestnut2017hardness}.
The goal there is typically to worsen the optimum value for a structure that an enemy is trying to build by attacking or interdicting edges or nodes of the graph. Examples of interdicted structures studied include matchings~\cite{Zen10}, minimum spanning trees~\cite{LS17,Zen15}, shortest paths~\cite{golden1978problem,israeli2002shortest}, $st$-flows~\cite{Philips93,Wood93,Zen10DAM} and global minimum cuts~\cite{Zen14}. ~\cite{chestnut2016interdicting} contains a good overview of the current literature.

Related to the problem of increasing the weight of the MST that we study, prior work~\cite{jalg99,tcs99} has addressed the problem of {\em decreasing} the total cost or bottleneck cost of the MST by paying upgrading costs on the nodes that in turn reduce the weight of the incident edges.

\section{Correctness of MST Upgrading}
\label{sec:correct}
\subsection{Introduction of the problem}

\iffalse
Consider the following upgrading problem:

\begin{problem} (MMST)
	Given a graph $G$, functions $w, c:E(G)\to \mathbb{N}$ and a budget $B$, find a downgrading scheme $x:E(G)\to \mathbb{N}$ that maximizes $f_x$ the weight of an MST in $(G,w+x)$ while ensuring the total cost, $c_x := c \cdot x$, is at most $B$. We call $f_G(B)$ this maximum value.
\end{problem}

We can also consider the targeted version of this problem:
\begin{problem} (targeted MMST)
	Given a graph $G$, functions $w, c:E(G)\to \mathbb{N}$ and a target $T$, find a downgrading scheme $x:E(G)\to \mathbb{N}$ that minimizes the cost $c_x := c \cdot x$ while ensuring that $f_x$ the weight of all MSTs in $(G,w+x)$, is at least $T$.
\end{problem}

\begin{problem} (Continuous MMST)
	Given a graph $G$, functions $w, c:E(G)\to \mathbb{N}$ and a budget $B$, find a upgrading scheme $x:E(G)\to \mathbb{R^+}$ that maximizes $f_x$ the weight of an MST in $(G,w+x)$ while ensuring the total cost, $c_x := c \cdot x$, is at most $B$. We call $f_G(B)$ this maximum value.
\end{problem}

Frederickson and Solis-Oba study the above continuous version of the Budgeted MMST problem \cite{FS96}, where the key difference is that the upgrading amounts on the edges are allowed to be fractional rather than integral values.

\fi
Frederickson and Solis-Oba studied the continuous version of the Budgeted MMST problem \cite{FS96}, and  prove that an optimal solution can be found in polynomial time. However, there is an error in the proof of the optimality of their algorithm. We emphasize that their algorithm is correct. In this section, we point out the mistake and provide a correct proof. 

Let us first introduce the vocabulary used in~\cite{FS96}. Given $S \subset E$,
let $c(S) = \sum_{e \in S} c(e)$.
We define $coverage(S,G)$ as the minimum number of edges that any minimum spanning tree (under the current weights $w$) of $G$ shares with $S$.
We say that a set $S$ is \textit{lifted} by $\delta$ when the weight of every edge in $S$ is increased by the same amount $\delta$. 
Let $tolerance(S, G)$ be the maximum amount that the weights of the edges in $S$ can be lifted by until $coverage(S, G)$ changes (It can be shown that the coverage will only decrease at this change).  
Define $inc\_cost(S, G) := c(S)/coverage(S, G)$, the cost per unit increment of $S$.
Roughly, the algorithm in~\cite{FS96} is greedy and chooses a set $S$ of minimum $inc\_cost$ value in the current weighted graph and lifts the weights of all the edges in it by its $tolerance$ or until it runs out of budget (whichever occurs first).

Let $\widetilde{G_{w_i}}$ be the graph obtained from $G$ by first deleting all edges of weight strictly larger than $w_i$ and then contracting all edges of weight strictly smaller than $w_i$. For an edge $e\in E(G)$, let $sm\_eq(e, G)$ be the set of edges in $G$ whose weights are at most the weight of $e$. 

We first point out the following observation. Let $T$ be an MST of $G$ and let $T_{w_i}$ be the set of edges of weight $w_i$ in $T$. Then, it is easy to check that $T_{w_i}$ forms a spanning forest in $\widetilde{G_{w_i}}$. Thus, if $S$ is a set of edges in $G$ with the same weight $w_i$, then $coverage(S,G)$ is equal to the increase in the number of components of $\widetilde{G_{w_i}}$ after deleting $S$. 

\iffalse++++++++++++++++++++++++++++++++
\begin{lemma} 
\label{Coverage}
If $S \subset E$ with all edges of $S$ having the same weight $w_i$, then $coverage(S,G) = comps(S,\widetilde{G_{w_i}})$, where $comps(A,H)$ is the number of additional components created by deleting $A$ from $H$.
\end{lemma}
++++++++++++++++++++++++++++++++++++++++++++++++++++++++++
\fi

\subsection{Continuous Downgrading of MSTs}

In \cite{FS96}, Frederickson and Solis-Oba proposed the following algorithm to maximize the weight of the MSTs in the resulting graph. 
Note that when there is a choice of sets $S$ that minimize $inc\_cost$, the algorithm is free to choose any.

\begin{algorithm}[h]
\begin{algorithmic}[1]
\STATE $balance \gets B$; 
\STATE $wmst$ $\gets$ weight of a minimum spanning tree of $G$; 
\WHILE {$balance > 0$}
	\STATE Find a set $S$ that minimizes $inc\_cost(S, G)$; 
	\STATE $A \gets min\{tolerance(S,G),balance/c(S)\}$;
	\STATE Lift the weights of the edges in $S$ by $A$; 
	\STATE $balance \gets balance - A * c(S)$;
\ENDWHILE
\STATE $increase \gets$ (weight of a minimum spanning tree of $G$)$- wmst$; 
\STATE Output $increase$; 
\end{algorithmic}
 \caption{raise\_mst $(G,w,c,B)$}
\end{algorithm}

Unfortunately, the proof of optimality of the algorithm (Theorem 3.1 in \cite{FS96}) is incorrect. The main idea of their proof is to first decompose any given optimal solution into a sequence of fractional lifts (where a lift corresponds to the action of increasing a set of edges 
%of identical weight 
by the same amount). Then, at any partial budget $b<B$, one can identify $S_b^*$, the set of arcs being lifted after spending $b$ budget according to the decomposition of the optimal solution. Then, one can compare $S_b^*$ to $S_b$, the set of arcs lifted by $raise\_mst$ after spending $b$ budget. 
The error in their proof arises in their decomposition wherein their definition of $S_b^*$ is not a lift. 
We explain why this is important in the rest of the proof next.
As a last step in their proof, they show that $S_b$ performs as well as $S_b^*$ for any $b$ by using the following lemma (Lemma 3.1 in \cite{FS96}). 
%*** Can we talk about the error right away here? Maybe bring the definitions of the star sets here? *****

\begin{lemma}[Lemma 3.1 in \cite{FS96}]
    \label{lem:sumeq}
    Let $G, G'$ be two graphs on the same vertex and edge sets but has edge-weights $w, w'$ respectively. Let $S$ be a set of edges such that for every $e\in S$, $sm\_eq(e, G')\subseteq sm\_eq(e, G)$. Then, then $coverage(S, G)\le coverage(S, G')$. 
\end{lemma}

Given the decomposition, let $G_{b}^*$ and $G_b$ be the state of the graph after spending budget $b$ according to the decomposition of the optimal solution and $raise\_mst$ respectively. Let $w_{b^*}, w_b$ be the weight function of the edges in $G_b^*$ and $G_b$ respectively. In the discussion following their proof of Lemma \ref{lem:sumeq}, they pointed out that if it is possible to guarantee that the decomposed $S_b^*$ only includes edges $e$ for which $sm\_eq(e, G_b)\subseteq sm\_eq(e, G_b^*)$, then by Lemma \ref{lem:sumeq}, the following inequality holds:
\begin{equation}
\label{wrongineq}
inc\_cost(S^*_b, G_b)\le inc\_cost (S^*_b, G^*_b).
\end{equation}

Since $S_b$ is chosen by $raise\_mst$ to be an optimal set to lift, it follows that $raise\_mst$ is as good as an optimal solution. More specifically, between budget $b$ and $b+\epsilon$ for some $\epsilon>0$, spending $\epsilon$ to uniformly lift the set $S^*_b$ increases the MST in $G_b^*$ by $\frac{\epsilon}{inc\_cost (S^*_b, G_b^*)}$. From inequality~\eqref{wrongineq}, this value is at most $\frac{\epsilon}{inc\_cost (S^*_b, G_b)}$, which is at most $\frac{\epsilon}{inc\_cost (S_b, G_b)}$ due to the choice of $S_b$ by $raise\_mst$. Hence, spending $\epsilon$ on $S_b$ in $G_b$ is as worthwhile as spending it on $S_b^*$ in $G_b^*$. However, this argument relies on the subtle fact that $S_b^*$ is lifted by the same amount $\epsilon/c(S^*_b)$ on every edge. We will see that in the proposed decomposition of the optimal solution in \cite{FS96}, the $S_b^*$'s are not lifting sets in this way. 

%In particular, they define $G_b$ to be the state of the graph after spending a partial budget $b$ and define $S_b$ as the set of arcs the algorithm decides to lift at that time. Then, they analogously define $s_b^*$ and $G_b^*$ for the decomposition of the optimal solution. This chain of lifts $S_b$ and $S_b^*$ can be used to compare amount of increase in the MST by $raise\_mst$ to the optimal amount. Unfortunately, what they propose as a decomposition of their optimal solution is not a series of lifts. This poses a problem, as then we cannot compare directly $inc\_cost(S_b,G_b)$ and $inc\_cost(S_b^*,G_b^*)$, since those measures are based on the assumption that we are lifting every edge of $S_b^*$ by the same amount. The following is an example where their decomposition is not a series of lifts.

%In $raise\_mst$, a sequence of edge-sets are lifted. Instead of viewing it as discretely lifting a set by $\Delta$ amount, we can interpret it as a continuous process where budgets are continuously spent to raise the set until we raised it by $\Delta$. Then, define $S_b$ as the set of arcs being lifted after spending budget $b$. Let $G_b$ be the graph after spending budget $b$ according to $raise\_mst$ where $w_b$ is the weight function of the edges.

Here is the definition of $S_b^*$ in \cite{FS96}: Given $0<b<B$, for $\Delta>0$, let $c_\Delta$ be the total cost of bringing all edges from its initial weight of $w(e)$ to $\min\{w^*(e), w_b(e)+\Delta\}$. Let $\Delta_b$ be such that $c_{\Delta_b}=b$ and let $w^*_b(e)=\min\{w^*(e), w_b(e)+\Delta_b\}$ for all edges $e$. Select $S^*_b$ to include all of the edges $e$ for which $w_b^*(e)=w_b(e)+\Delta_b \le w^*(e)$ and $w^*_{b-\epsilon}(e)< w^*(e)$ for all $b\ge \epsilon >0$. 

\paragraph{Counterexample showing $S^*_b$ is not a lift.} 
%**** ADD triangle picture here?

Consider a path with three edges $e_1$, $e_2$ and $e_3$ whose initial weights are $0$. Assume the unit cost of raising the weight of any edge is $1$. For $B=4$, let $w^* = (1,2,1)$ and $w = (2,2,0)$ be their final weight in an optimal solution and $raise\_mst$ respectively. Furthermore, assume $raise\_mst$ simply found the set $\{e_1, e_2\}$ and raised its weight by $2$ until all the budget is spent. Following the definition in \cite{FS96}, $S_b = \{1,2\}$ for all $0\le b\le 4$. Recall that $w^*_b = min(w^*,w_b + \Delta_b)$ with $\Delta_b$ such that the cost of increasing to these weights would be $b$.

Then $w^*_2 = (1,1,0)$. When $b=2+\epsilon$, $w_b=(1+\epsilon/2, 1+\epsilon/2, 0)$. Due to the cap imposed by $w$, $w^*_{2+\varepsilon} = (\min(1, 1+\epsilon/2+\Delta_{2+\epsilon}),\min(2, 1+\epsilon/2+\Delta_{2+\epsilon}), \min(1, 0+\Delta_{2+\epsilon})) =(1, 1+\epsilon/2 +\Delta_{2+\epsilon}, \Delta_{2+\epsilon})$. Thus, by setting $\Delta_{2+\epsilon}=\epsilon/4$, we achieve a weight of $w_{2+\epsilon}= (1, 1+\frac{3 \varepsilon}{4}, \frac{\varepsilon}{4})$, spending a total budget of $2+\epsilon$. We see that the lifting set $S^*_{2+\varepsilon}$ should be $\{e_2,e_3\}$. However, $e_2$ and $e_3$ are not lifted by the same amount. Thus, it is unclear how much improvement was achieved by the optimal solution in the range of $(w, w+\epsilon)$, making it incomparable to $raise\_mst$. 
%*** Better counterexample ??? ***

\begin{remark}
We notice that an analogous proof is given in \cite{FS98} in section 3.1 in the context of matroids. This proof contains the same problem, and the following correction can also be applied.  
\end{remark}

\paragraph{A Corrected Decomposition.}

We now provide a valid decomposition of an optimal solution, leading to a correct proof of the same result.

\iffalse__________________________________7
We require the following lemma from \refFS96{} to complete the proof. Let $sm\_eq(S,G)$ be the subset of edges in $G$ of weight less or equal to $e$.

\begin{lemma} \label{lem:sumeq}
Let $G^* = (V,E,w^*)$ and $G = (V,E,w)$ be undirected graphs and $S^* \subset E$. If for every $e \in S^*$, $sm\_eq(e, G) \subset sm\_eq(e, G^*)$ then $coverage(S^*,G^*) \leq coverage(S^*,G)$
\end{lemma}
\begin{proof}
The idea of the proof is that since $sm\_eq(e, G) \subset sm\_eq(e, G^*)$, an MST in $G$ cannot avoid the increase from $S^*$ by going to a smaller edge. \\
Indeed, if $T$ is an MST of $G$, then we can modify it into an MST of $G^*$, by removing the largest edge in the cycle created by $T \cup \{e\}$, for every $e\in E$.
We remark in particular that if $f\in S^* \setminus T$, then $w(f) \leq w(e)$, but by hypothesis, since $e\in S^*$, this means that $w^*(f) \leq w^*(e)$. Thus we can construct an MST in $G$ that contains no edge of $S^* \setminus T$. This means that $coverage(S^*,G^*) \leq coverage(S^*,G)$.
\end{proof}

_______________________________7
\fi

%Given an optimal solution, we decompose it into a sequence of lifts. Then, viewing both the optimal and $raise\_mst$ as a continuous process, we can compare the lifted sets after spending any fractional budget $b$ and show that $raise\_mst$ does as well as the optimal solution. Thus, it is useful to identify the edge weights and the lifted edge-sets after spending any budget $b$. 

For any real number $0\le b\le B$, let $S_b$ denote the edge set lifted by $raise\_mst$ after spending budget $b$. Let $G_b = (V,E,w_b)$, the state of the graph produced by the algorithm at budget $b$. More precisely, $S_b$ is the set of edges lifted in the $j$-th iteration where $j$ is the largest integer such that the budget spent at the beginning of the iteration is less or equals to $b$. Note that $raise\_mst$ lifts an edge set until it reaches its tolerance. Since all initial weights are integral, the tolerance of a set can only change once its weight were increased to the next integral value.

Then, let $b_0=0$, $b_1, ..., b_{k-1},b_k=B$ be the sequence of budget values that the algorithm $raise\_mst$ spends such that between any consecutive values, the difference was spent by $raise\_mst$ to lift the weight of an edge set to the next integral amount (or until out of budget). This implies that after spending budget $b_i$, an additional $b_{i+1}-b_i$ budget is used to increase the weight of an edge set $S_{b_i}$ by one unit, or possibly less because we ran out of budget. Let $w_b(e)$ denote the weight of edge $e$ after spending budget $b$. Then for $e\in S_{b_i}, e'\notin S_{b_i}$, the weight of $e'$ does not change as budget increases from $b_i$ to $b_{i+1}$, i.e. $w_{b_i}(e')=w_{b_{i+1}}(e')$. Furthermore, since $w_{b_i}(e')$ is an integer, it is not strictly between the interval $(w_{b_i}(e), w_{b_{i+1}}(e))$.

%We can decompose this into a series of sets that are lifted : this gives a procedure $\mathcal{A}$.
%We call $S_b^*$ the set lifted when at budget $b$ by $\mathcal{A}$.  We call $G_b^*= (V,E,w^*_b)$ the state of the graph at the partial solution produced by the algorithm at budget $b$.

%We will be choosing the sets $S_b^*$ in order to easily prove the optimality of \textit{raise\_mst}, by showing that $inc\_cost(S_b,G_b) \leq inc\_cost(S_b^*,G_b^*)$ for almost all $b$.

Given an optimal solution, let $w^*$ represent the final weights of the edges. On a high level, we will decompose this optimal solution into a sequence of lifts by following $raise\_mst$ as closely as possible. At some budget $b$, an edge might reach its final weight $w^*$ and we must find an alternative way to spend the excess budget. Thus, we raise the rest of the edges in $S_b$ faster, until they reach what they are supposed to be relative to how much $raise\_mst$ has spent its budget. Lastly, if any excess budget remains, it is spent on lifting all edges (that are not yet capped) by the same amount $\Delta$. Then, the weight of any edge after spending budget $b$ in our decomposition has two components, $w_b$, the amount from emulating $raise\_mst$, and a global increase of $\Delta$.

In order to explicitly define the weight of an edge in our decomposition after spending budget $b$, let $c_{\Delta,b}(e)$ be the cost to increase an edge $e$ with initial weight $w(e)$ to $min(w^*(e), w_b(e) + \Delta)$. Lightly abusing the notation, let $c_{\Delta, b}=\sum_{e\in E(G)} c_{\Delta, b}(e)$. Note that when $\Delta=0$, due to the cap imposed by $w^*$, we might no longer need to spend all of the budget $b$ and thus $c_{0, b}\le b$. 
Also, when there are no caps imposed by $w^*$, we have $c_{0, b} = b$.

Ideally, we would like $c_{\Delta, b} =b$ because it is easier compare how the optimal and $raise\_mst$ spend the budget. Thus, for $0\le i < k$, let $\Delta_i$ be a value such that $c_{\Delta_i,b_i} = b_i$. Informally, since some edges get capped by $w^*$, simply raising the weight of edges to $w_{b_i}$ might cost less than $b_i$. Then, $\Delta_i$ represents how much extra global weight we have to increase every non-capped edge in order to match the spending of $b_i$.

First, we show the following claim:

\begin{claim}
$\{\Delta_i\}_{i=0}^{k-1}$ is a non-decreasing sequence.
\end{claim}

\begin{proof}
We will show that $\Delta_{i+1}\ge \Delta_i$. Consider the difference between $b_i= c_{\Delta_i, b_i}$ and $c_{\Delta_i, b_{i+1}}$. To achieve the second cost, the weight of some of the edges is increased to $w_{b_{i+1}}$. Since some of edges might get capped, the change in weight is at most $w_{b_{i+1}}-w_{b_i}$. Since $raise\_mst$ spends $b_{i+1}-b_i$ to change the edges from weight $w_{b_i}$ to $w_{b_{i+1}}$, the difference in the two above costs is at most $b_{i+1}-b_i$. Then, $c_{\Delta_i, b_{i+1}}\le c_{\Delta_i, b_i} + (b_{i+1}-b_i) = b_{i+1} = c_{\Delta_{i+1}, b_{i+1}}$. Thus, it follows that $\Delta_{i+1}\ge \Delta_i$.
\end{proof}

Note that after spending budget $b_i$, we should expect our decomposition of the optimal solution to ensure each edge has weight $\min\{w^*(e), w_{b_i}(e)+\Delta_i\}$. 
Then, between a budget spending of $b_i$ and $b_{i+1}$,the weight of an uncapped edge in the decomposition should grow from $w_{b_i}(e)+\Delta_i$ to $w_{b_{i+1}}(e)+\Delta_{i+1}$. To achieve this 
increase, we break each interval $[b_i, b_{i+1}]$ of the budget-spending process into two phases. In the first phase, we spend the increase in budget to augment the weight of edges from $w_{b_i}+\Delta_i$ to $w_{b_{i+1}} + \Delta_i$. 
Then, in the second phase, we increase their weight from $w_{b_{i+1}}(e)+\Delta_i$ to $w_{b_{i+1}}(e)+\Delta_{i+1}$. Let $\beta_i = c_{\Delta_i,b_{i+1}} $, representing the transitioning point (in terms of the budget) between the two phases.

We then define $w^*_b$. 

\begin{definition}
The edge weight $w^*_b(e):=$
\begin{itemize}
\item $min(w^*(e), w_{f(b)}(e) + \Delta_i)$, with $f:\mathbb{R}\to \mathbb{R}$ such that $c_{\Delta_i,f(b)} = b$, for $b \in [b_i,\beta_i]$,
\item $min(w^*(e), w_{b_{i+1}} + \Delta(b))$, with $\Delta(b)$ such that $c_{\Delta(b),b_{i+1}} = b$, for $b \in [\beta_i,b_{i+1}]$,
\end{itemize}
\end{definition}

\begin{claim}
\label{cl:fdelta}
The following holds:
    \begin{enumerate}
        \item $f(b_i)=b_i$ and $f(\beta_i)=b_{i+1}$,
%        \item $w^*_{b_i}(e)=\min(w^*(e), w_{b_i}(e)+\Delta_i)$$
        \item $f(b)$ is an increasing function in the interval $[b_i, \beta_i]$
        \item $\Delta(\beta_i) = \Delta_i$ and $\Delta(b_{i+1}) = \Delta_{i+1}$, and
        \item $\Delta(b)$ is an increasing function in the interval $[\beta_i, b_{i+1}]$.
    \end{enumerate}
\end{claim}

\begin{proof}
The first statement follows from the definition of $\Delta_i$ and $\beta_i$. To prove the second statement, consider $b, b'$ such that $b_i\le b<b'\le \beta_i$. 
Note that $c_{\Delta_i, f(b)}$ and $c_{\Delta_i, f(b')}$ corresponds to the cost of raising edges to a weight of $\min\{w^*(e), w_{f(b)}+\Delta_i\}$ and $\min\{w^*(e), w_{f(b')}+\Delta_i\}$. 
Since by definition, $c_{\Delta_i, f(b)} = b <b' = c_{\Delta_i, f(b')}$, it follows that there exists some edge $e$ such that $w_{f(b)}(e)<w_{f(b')}(e)$. 
Since $w_{b''}(e)$ is non decreasing with respect to $b''$, it follows that $f(b)< f(b')$, proving the second statement. 

The third statement follows from the definition of $\beta_i$ and $\Delta_{i+1}$. For the fourth statement, we use similar arguments as the proof for the second statement. Consider $b, b'$ such that $\beta_i\le b<b'\le b_{i+1}$. Since $c_{\Delta(b), b_{i+1}}$ corresponds to the cost of raising edges to a weight of $\min\{w^*(e), w_{b_{i+1}} +\Delta(b)\}$, does not increase if $\Delta(b)$ does not increase. Then, to achieve a higher cost of $b'$, it follows that $\Delta(b')< \Delta(b)$. 

%This ensures that the weight of an edge $e$ after spending budget $\beta_i$ and $\Delta_{i+1}$ is the desirable value $w^*_{b_i}(e)=\min(w^*(e), w_{b_{i+1}}(e)+\Delta_i)$ and $w^*_{b_i}(e)=\min(w^*(e), w_{b_{i+1}}(e)+\Delta_{i+1})$ respectively.

\end{proof} 

\begin{remark}
If we extend the definition of $f(b)$ and $\Delta(b)$ to $f(b)=b$ for $b\in [\beta_i, b_{i+1}]$ and $\Delta(b)=\Delta_i$ for $b\in [b_i, \beta_i]$, then $f(b)$ and $\Delta(b)$ become two non-decreasing continuous functions.  Furthermore, we have that $w^*_{b}(e) = min(w^*(e), w_{f(b)}(e) + \Delta(b))$ and $w^*_b(e)$ is also a non-decreasing function. It is also clear that after spending a total budget of $B$, all edges should reach its desired maximum weight of $w^*$.
\end{remark}

Here is an intuitive explanation of $f(b)$ and $\Delta(b)$. During the first phase, the decomposition would like to copy $raise\_mst$ and lift all the edges in $S_{b_i}$. However, at some point, an edge $e\in S_{b_i}$ might have reached its full capacity $w^*$. Then, in order to match the same amount of spending as $raise\_mst$, we would need to raise the rest of the edges in $S_{b_i}$ at a faster rate. Thus, to actually spend budget $b$, we had to lift some edges higher, to a point where if all edges in $S_{b_i}$ were raised to that point would have costed us $f(b)$. Similarly, $\Delta(b)$ is an adjustment function, corresponding to a faster rate of increasing $\Delta$ in the second phase, caused by some edges reaching its full capacity.

\paragraph{Example of the Correct Decomposition}: Consider performing the above decomposition on our previous counter-example. Recall that $G$ is a path of three edges whose initial weights are $0$. With a total budget of $4$, $raise\_mst$ lifts $e_1, e_2$ by $2$ to a final weight of $(2, 2, 0)$ for edges $(e_1, e_2, e_3)$. An optimal solution has a final weight of $(1, 2, 1)$. By definition, $b_0=0, b_1=2, b_2=4$ where the set $S=\{e_1, e_2\}$ is first raised to weight $1$ then raised to weight $2$. Since $c_{0, b_0}=0=b_0$, $\Delta_0=0$. Between $b_0$ and $b_1$, our decomposition would copy $raise\_mst$ exactly, so $f(b)=b, \Delta(b)=0$ for $b_0\le b\le b_1$. Note that $\beta_0=b_1$ and $\Delta_1=0$ since there is no second phase in this interval. See Figure 1. %\ref{examplegraph}. 

Now, consider when $b\in [b_1, b_2]$. Note that $c_{1, b_2}$ is the cost of raising the weights to $(\min(1, 3), \min(2, 3), \min(1, 1)) = (1, 2, 1)$. Then, $c_{1, b_2}=4 = b_2$ and therefore $\Delta_2=1$. Note that $\beta_1=c_{\Delta_1, b_2}$ is the cost of reaching a weight of $(1, 2, 0)$ and thus $\beta_1=3$. In the first phase, since $e_1$ is already capped, $c_{0, b'}$ is the cost of reaching a weight of $(1, b'/2, 0)$. Therefore $c_{\Delta_1, b'} = b'/2+1$. In $raise\_mst$, after spending budget $b$ to ensure $c_{\Delta_1, f(b)}=b$, $f(b)=2b-2$. Then, $w^*_b=\min(w^*, w_{f(b)}) = (1, b-1, 0)$, corresponding to spending the current increase in budget to raise $e_2$'s weight from $1$ to $2$. In the second phase, for $3=\beta_1 \le b\le b_2=4$, note that $c_{\Delta', b_2}$ is the cost of raising the weights to $(1, 2, \Delta')$. Then, $c_{\Delta', b_2}= 3+\Delta'$. To ensure $c_{\Delta(b), b_2}= b$, we see that $\Delta(b) = b-3$. Then, $w^*_b= (1, 2, b-3)$ corresponding to lifting the weight of $e_3$ from $0$ to $1$. Note that now each step of the decomposition corresponds to a proper lift. More precisely, between $[b_1, \beta_1]$, only edge $e_2$ is lifted and between $[\beta_1, b_2]$ only edge $e_3$ is lifted.

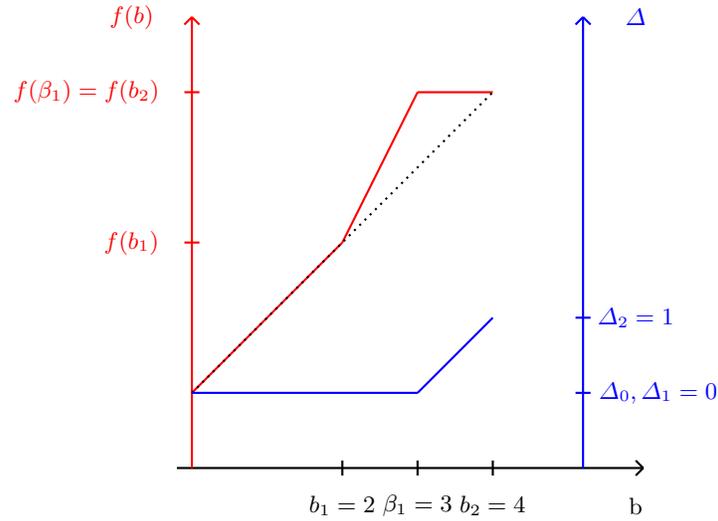
\begin{figure}[h]
\begin{center}
\begin{tikzpicture}[thick]
    
    \draw (-0.2,-1) to (6,-1);
	\draw (2,-1.1) to (2,-0.9);
	\draw (3,-1.1) to (3,-0.9);
	\draw (4,-1.1) to (4,-0.9);
	\draw (5.9, -0.9) to (6, -1);
	\draw (5.9, -1.1) to (6, -1);
	
	\node(x0) at (5.9, -1.5) {b};
	\node (X1) at (2,-1.5) {$b_1=2$};
	\node (Y2) at (3,-1.5) {$\beta_1=3$};
	\node (Z3) at (4,-1.5) {$b_2=4$};
	
	\draw[red] (0,-1) to (0,5);
	\draw[red] (-0.1,2) to (0.1,2);
	\draw[red] (-0.1,4) to (0.1,4);
	\draw[red] (-0.1, 4.9) to (0, 5);
	\draw[red] (0.1, 4.9) to (0, 5);

	\draw[red] (0,0) to (2,2);
	\draw[red] (2,2) to (3,4);
	\draw[red] (3,4) to (4,4);
	
	\node[red] (X) at (-0.8,2) {$f(b_1)$};
%	\node[red] (Y) at (-0.8,3) {$f(\beta_1)$};
	\node[red] (Z) at (-1.4,4) {$f(\beta_1)=f(b_2)$};
	\node[red] (y4) at (-0.8, 5) {$f(b)$};
	
	\draw[dotted] (0,0) to (4,4);
		
	\draw[blue] (5.2,-1) to (5.2,5);
	\draw[blue] (5.3,0) to (5.1,0);
	\draw[blue] (5.3,1) to (5.1,1);	
	\draw[blue] (5.1, 4.9) to (5.2, 5);
	\draw[blue] (5.3, 4.9) to (5.2, 5);
	
	\draw[blue] (0,0) to (3,0);
	\draw[blue] (3,0) to (4,1);
	
	\node[blue] (D) at (5.9,1) {$\Delta_2 = 1$};
	\node[blue] (D) at (6.2,0) {$\Delta_0,\Delta_1 = 0$};
	\node[blue] (y5) at (5.9, 5) {$\Delta$};
	
\end{tikzpicture}
\end{center}
\label{examplegraph}
\caption{The functions $f$ (in red) and $\Delta$ (in blue) in the case of the previously discussed counter-example.}
\end{figure}

\paragraph{The new decomposition provides lifts}:
We now show that if we view the optimal solution as a continuous process that increases the weights of edges according to $w^*_b$, then the decomposition produces a sequence of lifts. For any $0\le b<B$, let $S^*_b$ be the set of edges whose weight changes at budget $b$ according to $w^*_b$. Formally, $S^*_b=\{e: w^*_b(e)< w^*_{b+\epsilon} (e) \forall \epsilon>0\}$. \\

However, when lifting from $b_i$ to $\beta_i$ (or from $\beta_i$ to $b_{i+1}$),  we may however be lifting different sets. Indeed, we could hit one of the upper bounds for one of the edges that we are lifting during this period, and so the set we are lifting would shrink to another, smaller set. This is why we now introduce new breakpoints $(p^i_j)$ in between the $b_i$s and $\beta_i$s, in order to always have proper lifts, which will guarantee that we can indeed compare the optimal solution to the solution given by $raise\_mst$.

\begin{lemma}
For every $0\le i< k$,  there exists a finite sequence $\{p^i_j\}^l_{j=0}$ where $p^i_0=b_i$, $p^i_l=b_{i+1}$ and there exists $0\le k\le l$ such that $p^i_k=\beta_i$. Furthermore, for any interval $[p^i_j, p^i_{j+1}]$, the set of edges whose weights changed according to $w^*_b$ when $b$ is within this interval corresponds to a proper lift. More precisely, for all $p^i_j \le b < p^i_{j+1}$, we have that $S^*_b=S^*_{p_j}$ and $w^*_{b}(e)-w^*_{p_j}(e) = w^*_{b}(e')-w^*_{p_j}(e')$ for any $e, e'\in S^*_b$. 
\end{lemma}

\begin{proof}
Fix $0\le i < k$. We find the sequence separately for each of the two phases. Suppose we are in the first phase where $b_i\le b\le \beta_i$. Let $b_i < p_1<p_2<...<\beta_i$ be a sequence of budgets at which some edge reaches its cap. Formally, $p_j$ is in the sequence if there exists an edge $e$ such that $w^*_{p_j-\epsilon}(e) < w^*_{p_j}(e) = w^*(e)$ for all $\epsilon>0$. Note that since there are only finitely many edges and each edge can reach its maximum cap only once, this sequence is finite. We claim that within the interval $(p_j, p_{j+1})$, edges are being properly lifted with respect to $w_b^*$. 

Let $p_j\le b<b'\le p_{j+1}$. Let $P^*=\{e: w^*_{p_j}(e) = w^*(e)\}$, representing the set of all edges that reached its maximum cap at budget $p_j$. First we show $S^*_b=S_{b_i} \setminus P^*$. Let $e\in S^*_{b}$. Since no edges becomes capped between $p_j$ and $p_{j+1}$, it follows that $w^*_{b} = w_{f(b)}(e)+\Delta_i$ and $w^*_{e}=w_{f(b')}+\Delta_i$. By definition of $S^*_{b}$, it follows that $w^*_b<w^*_{b'}$. From Claim \ref{cl:fdelta}, we know $b_i\le f(b)< f(b')\le b_{i+1}$. Then, $w_{b_i}(e) \le w_{f(b)}(e) < w_{f(b')}(e) \le w_{b_{i+1}}(e)$. Thus, the weight of edge $e$ changed with respect to $w_b$, proving $e\in S_{b_i}$. Since $w^*_b(e)<w^*_{b'}(e)$, it also follows that $e$ does not become capped at nor before $b$, proving $e\notin C^*$. Thus, $S_b^*\subseteq S_{b_i}\setminus C^*$. 

For the other direction, let $e'\in S_{b_i}\setminus C^*$. Since $e'$ does not get capped, $w^*_b(e')= w_{f(b)}(e')+\Delta_i$ and $w^*_{b'}(e')= w_{f(b')}(e')+\Delta_i$. 
From Claim \ref{cl:fdelta}, since $b_i\le f(b)<f(b')\le b_{i+1}$, it follows that $w_{f(b)}(e')<w_{f(b')}(e')$ and thus $w^*_b(e)< w^*_{b'}(e)$. 
Since this inequality holds for any $b'<b''<b_{i+1}$, it follows that $w^*_b(e)<w^*_{b+\epsilon}(e)$ for all $\epsilon>0$, proving $e\in S^*_b$. Thus, we conclude that $S^*_b=S^*_{p_j}=S_{b_i}\setminus P^*$.

Let $e, e'\in S^*_b$. Since $e$ is not capped, $w^*_b(e)-w^*_{p_j}(e) = w_{f(b)}(e)-w_{f(p_j)}(e)$. Since $e, e'\in S^*_b\subseteq S_{b_i}$, $raise\_mst$ lifts their weight by the same amount in the interval $[b_i, b_{i+1}]$. Since $b\in [b_i, b_{i+1}]$, it follows that $w_{f(b)}(e)-w_{f(p_j)}(e) = w_{f(b)}(e')-w_{f(p_j)}(e')$, proving our lemma holds for the first phase.

For the second phase of the interval when $\beta_i\le b< \beta_{i+1}$, let $\beta_i < q_1<q_2<...<b_{i+1}$ be a sequence of budgets such that some edge $e$ becomes capped at budget $q_j$. Once again, this sequence is finite. Let $q_j\le b<q_{j+1}$. Since no edges becomes capped, $w^*_b(e)=w_{b_{i+1}}(e)+\Delta(b)$ holds for any $e\in S^*_b$. From Claim \ref{cl:fdelta}, since $\Delta(b)$ is an increasing function, it follows that $w^*_b$ increases for all non-capped edge $e$. Since no new edges become capped, $S^*_b=S^*_{q_j}$. It also follows that for any edge $e\in S^*_b$, $w^*_b(e)-w^*_{q_j}(e)= \Delta(b)-\Delta(q_j)$, proving the lemma also holds for the second phase. 

%After spending budget $\beta_i$, for small $\epsilon>0$, the weight of all uncapped edges will increase by the same amount $\Delta(\beta_i+\epsilon)-\Delta(\beta_i)$ provided no edge becomes capped. Thus, if we define $S^*_{b}=\{e\in E(G): w^*_{b1}(e)< w^*(e)\}$, then $S^*_{\beta_i}$ represents a set being properly lifted after spending budget $\beta_i$. This set is lifted until an edge becomes capped, which then the remaining uncapped edges forms the new set and continue being lifted by the same amount. Similar to before, if $b<b_{i+1$, then more budget can still be spend to increase the weight of some edge. Thus, these lifting sets are well-defined and properly decomposes the optimal solution.
\end{proof}

Let $P^*_b=\{e: w^*_b(e) = w^*(e)\}$, representing the set of edges that reached its cap at budget $b$. The next corollary follows from the proof of the previous lemma.

\begin{corollary}
\label{cor:lift}
For an interval $[b_i, b_{i+1}]$, given the finite sequence $\{p^i_j\}$ from above, we can explicitly describe $S^*_{p^i_j}$. In particular:
\begin{itemize}
    \item if $b_i\le p^i_j<p^i_{j+1} \le \beta_i$, then $S^*_{p^i_j} = S_{b_i}\setminus P^*_{p^i_j}$,
    \item if $\beta_i\le p^i_j< p^i_{j+1}\le b_{i+1}$, then $S^*_{p^i_j} = E(G^*)\setminus P^*_{p^i_j}$. 
\end{itemize}
\end{corollary}

%During the $i$-th iteration, only edges in $S_{b_i}$ are lifted by $raise\_mst$. Then, from definition, during the first phase where $b_i\le b\le \beta_i$, only the eights of the edges in $S_{b_i}$ can potentially increase. Then, it follows from the definition of $w_b^*$ that all weights of the edges in $S_{b_i}$ are increased by the same amount until some edge hits its capacity. Then, the rest of the edges continue to be lifted by the same amount (but at a faster rate), until more edges hits its capacity. The process continues until a total budget of $\beta_i$ is spent. However, it is clear from definition that this process can be broken into finite number of discrete lifts. Similarly for the second phase. From the definition of $w^*_b$, the weight of all uncapped edges will increase by the same amount, until more edges hits its capacity, at which point all remaining edges continues to increase at a faster constant rate. The second phase can therefore also be decomposed into finite discrete lifts. This allow us to use Lemma \ref{lem:sumeq} and finish the proof. 

%$f(b)$ and $\Delta(b)$ are not the same units! Indeed, $f(b)$ is a certain spending, whereas $\Delta(b)$ is a certain amount of increase for a weight.Also, in the first phase, $S_b^* \subset S_b$. We also notice that this gives proper lifts.  

Now, we prove the correctness of the algorithm.

\begin{theorem}
The algorithm raise\_mst gives an optimal increase for continuous MMST.
\end{theorem}

\begin{proof}

The strategy of the proof is the same as the one in \cite{FS96}. 
We compare the lifted sets $S_{b_i}$ in $raise\_mst$ to the ones obtained from our decomposition, $S^*_b$. 
Given $0\le i< k$, let $\{p_k\}$ be the finite sequence defined in the previous lemma (we drop the superscript $i$ for convenience). We now show that for any interval $[p_j, p_{j+1}]$, lifting $S_{p_j}=S_{b_i}$ in $raise\_mst$ increases the MST of $G_{p_j}$ by at least as much as the increase in MST of $G^*_{p_k}$ caused by lifting $S^*_{p_j}$ in our decomposition. 
%Note that each phase can be decomposed into finitely many such intervals since our decomposition only changes the set to lift when an edge becomes capped, which happens at most finite amount of time. Splitting the interval if needed, we may assume $[b', b'']$ lies entirely in either the first phase or the second phase. Let $b'<b<b''$, then $S^*_b$ is the set being lifted in this interval. We break the analysis into two cases, depending on which phase the lift was performed in.

If the lift was performed in the first phase where $b_i \le p_j < p_{j+1} \le \beta_i$, then by Corollary \ref{cor:lift}, $S^*_{p_j}\subseteq S_{p_j}=S_{b_i}$. Let $e\in S^*_{p_j}, e'\in sm\_eq(e, G_{p_j})$. By Corollary \ref{cor:lift}, $e\in S_{b_i}$ and thus $w_{b_i}(e)\le w_{b}(e)<w_{b_{i+1}}(e)$ for all $b_i\le b<b_{i+1}$. We claim that $w_b(e')\le w_b(e)$ for all $b_i\le b \le b_{i+1}$. If $e'\in S_{b_i}$, since $raise\_mst$ only lifts edges of the same weight, $w_b(e')=w_b(e)$. If $e'\notin S_{b_i}$, then $w_{b}(e')=w_{b_i}(e')$. By our choice of $b_i$, $w_{b_{i+1}}(e)\le w_{b_i}(e)+1$. Then, $e'\in sm\_eq(e, G_{p_j})$ implies $w_{p_j}(e') \le \lfloor w_{p_j}(e)\rfloor =w_{b_i}(e)\le w_b(e)$, proving our claim.

It follows by definition of $w_{p_j}^*$ and $S_{p_j}^*$ that $w_{p_j}^*(e) = w_{f(p_j)}(e) + \Delta_i$ and $w_{p_j}^*(e') \leq w_{f(p_j)}(e') + \Delta_i$. Note that since $p_j<\beta_i$, $f(p_j) < b_{i+1}$ by Claim \ref{cl:fdelta}. Then, using the claim in the previous paragraph, $w_{f(p_j)}(e')\le w_{f(p_j)}(e)$, proving that $w^*_{p_j}(e')\le w^*_{p_j}(e)$. Therefore, $e'\in sm\_eq(e, G^*_{p_j})$ and $sm\_eq(e, G_{p_j}) \subset sm\_eq(e, G^*_{p_j})$.
Then, by Lemma \ref{lem:sumeq}, $coverage(S_b^*,G_b^*) \leq coverage(S_b^*,G_{f(b)})$, and

\begin{equation*}
\begin{split}
inc\_cost(S_{p_j},G_{p_j}) & \leq  inc\_cost(S_{p_j}^*,G_{p_j}) \quad \mbox{by the optimal choice of $S_{p_j}=S_{b_i}$} \\
& = \frac{c(S_{p_j}^*)}{coverage(S_{p_j}^*,G_{p_j})} \\
& \leq  \frac{c(S_{p_j}^*)}{coverage(S_{p_j}^*,G^*_{p_j})} = inc\_cost(S_{p_j}^*,G_{p_j}^*)
\end{split}
\end{equation*}

This implies after spending the same amount of budget $p_{j+1}-p_j$, $raise\_mst$ increased the MST value by at least as much as what the optimal solution did in their respective graphs.

Now, suppose the lift occurred in the second phase where $\beta_i\le p_j <p_{j+1}\le b_{i+1}$. Let $e\in S_{p_j}$ and $e'\in sm\_eq(e, G_{p_j})$. By Corollary \ref{cor:lift}, $e$ is an edge that is not capped after spending budget $p_j$. We claim that $w_{b_{i+1}}(e')\le w_{b_{i+1}}(e)$. If $e'\in S_{p_j}$, then $w_{b_{i+1}}(e')=w_{b_{i+1}}(e)$ since both edges are lifted together by $raise\_mst$. If $e'\notin S_{p_j}$, then $w_{b_{i+1}}(e')=w_{p_j}(e')\le w_{p_j}(e)\le w_{b_{i+1}}(e)$, proving our claim.

By definition of $w_b^*$ and $S_b^*$, it follows that $w_{p_j}^*(e) = w_{b_{i+1}}(e) + \Delta(p_j)$ and $w_{p_j}^*(e') \leq w_{b_{i+1}}(e') + \Delta(p_j)$. 
Thus it follows that $w^*_{p_j}(e') \leq w^*_{p_j}(e)$ and $w_b^*(e') \leq w_b^*(e)$. Then, by Lemma \ref{lem:sumeq}, we can conclude that $coverage(S_{p_j}^*,G_{p_j}^*) \leq coverage(S_{p_j}^*,G_{p_j})$. As before,

\iffalse++++++++++++++++++++++++++++++++++++++++++
\begin{claim}
If $w_{b}(e')\le w_{b}(e)$, then $w_{b_{i+1}}(e')\le w_{b_{i+1}}(e)$.
\end{claim}

If both $e'$ and $e$ are in $S_b$ or neither are in $S_b$, then the claim is true, since $raise\_mst$ either lifted them both equally or touched neither. If $e\in S_b$ but $e'\notin S_b$, then the claim is also clearly true, since $w_{b_{i+1}}(e) > w_{b}(e) \geq w_{b}(e') = w_{b_{i+1}}(e')$.

Lastly, if $e' \in S_b$ but $e \notin S_b$, then $raise\_mst$ is increasing $w_b(e')$ from $w_{b_i}(e')$ to $w_{b_{i+1}}(e')$. Since $b > \beta_i \geq b_i$, it follows from the definition of the $b_i$ sequence that $w_{b}(e) = w_{b_{i+1}}(e)$ is of a weight value at least as big as $w_{b_{i+1}}(e')$, proving our claim. 

+++++++++++++++++++++++++++++++++++++++
\fi

\begin{equation*}
\begin{split}
inc\_cost(S_{p_j},G_{p_j}) & \leq  inc\_cost(S_{p_j}^*,G_{p_j})  \quad \text{by optimality of $S_{p_j}$}\\
& =  \frac{c(S_{p_j}^*)}{coverage(S_{p_j}^*,G_{p_j})} \\
& \leq  \frac{c(S_{p_j}^*)}{coverage(S_{p_j}^*,G^*_{p_j})} = inc\_cost(S_{p_j}^*,G_{p_j}^*)
\end{split}
\end{equation*}

Then once again, by spending the same budget $p_{j+1}-p_j$, $raise\_mst$ increases the MST value by at least as much as the optimal solution. Then, it follows the total increase of MST performed by $raise\_mst$ is also optimal. 

\iffalse ___________________________________
And so, for $b\in[0,B] \setminus \cup \{\beta_i\}$, we have $inc\_cost(S_{b},G_{b}) \leq inc\_cost(S_b^*,G_b^*)$, that is, almost everywhere (for the Lebesgue mesure). Finally :
\begin{equation*}
\begin{split}
increase_{raise\_mst} & = \int_{0}^{B} coverage(S_b,G_b)\frac{db}{c(S_b)} \\ 
& =  \int_{0}^{B} \frac{db}{inc\_cost(S_b,G_b} \\
& \geq  \int_{0}^{B} \frac{db}{inc\_cost(S_b^*,G_b^*} \\
& = increase_{\mathcal{A}}
\end{split}
\end{equation*}
This concludes the proof.

___________________ \fi
\end{proof}

\section{Approximating Unit Perturbation MMST}
\label{sec:approx}

In this section, we present approximation algorithms for both the targeted and budgeted versions of MMST.

\subsection{A $2(1-\frac{1}{n})$-Approximation Algorithm for Targeted MMST}

\begin{theorem}
	Targeted MMST permits a $2(1-\frac{1}{n})$-approximation algorithm. 
\end{theorem}

\begin{proof}
The algorithm in \cite{FS96} computes in strongly polynomial time the function $f(B)$, the maximum weight of the minimum spanning trees of a graph attainable by spending a budget of value $B$ to increase the weights of its edges, potentially in fractional increments. Given a target $T \in \mathbb{N}$, let $B_T$ such that $f(B_T) = T$. Clearly, $B_T$ is a lower bound of the optimal value for (discrete) unit perturbation to reach MST target weight of $T$.

We will follow the algorithm in \cite{FS96} with budget $B_T$. Note that the algorithm lifts edges to their tolerance as long as there is enough budget to do so. Since the tolerances are integral, all edges are lifted an integral amount except for possibly the very last lift. Let $S$ denote the last set of edges computed by this algorithm  and \textit{balance} denote the remaining budget. The set $S$ induces a partition $P_1,\ldots, P_k$ of some graph $\tilde{G}_{w_l}$. In the last iteration, the algorithm in \cite{FS96} lifts all the edges in $S$ by $\frac{balance}{c(S)}$. The weight of the minimum spanning trees increases by $coverage(S,G) \frac{balance}{c(S)} \in \mathbb{N}$ since $T \in \mathbb{N}$. If $\frac{balance}{c(S)} \in \mathbb{N}$, then we are done. Otherwise, we first lift all the edges in $S$ by $\lfloor\frac{balance}{c(S)} \rfloor$. Then, it remains to increase the weight of the minimum spanning trees by

$$coverage(S,G) (\frac{balance}{c(S)} - \lfloor\frac{balance}{c(S)} \rfloor) < coverage(S,G) = coverage(S,\tilde{G}_{w_l}) = k-1.$$

The remaining budget $R$ is

\begin{equation}\label{eq:budget}
  R=balance - c(S) \lfloor\frac{balance}{c(S)} \rfloor = c(S) (\frac{balance}{c(S)} - \lfloor\frac{balance}{c(S)} \rfloor).  
\end{equation}

Assume that the node sets $P_i$ corresponding to the shores of the partition are ranked by increasing costs $c(\delta(P_i))$. Choose the $q = coverage(S,G) (\frac{balance}{c(S)} - \lfloor\frac{balance}{c(S)} \rfloor) $ cheapest shores $P_1,\ldots,P_{q}$ and lift all the edges in $\delta(P_1),\ldots,\delta(P_k)$ by one. Since \textit{tolerance}$(S,G) \ge \frac{balance}{c(S)} > \lfloor\frac{balance}{c(S)} \rfloor$, it follows that the weight of the minimum spanning trees increases by at least $q$. By the choice of the sets $P_1,\ldots,P_{q}$, we have

$$\frac{1}{q}\sum_{i=1}^qc(P_i) \leq \frac{1}{k}\sum_{i=1}^kc(P_i) = \frac{2c(S)}{k}.$$

This yields

\begin{align*}
	\sum_{i=1}^qc(P_i) &\leq \frac{2q}{k} c(S)\\
	                    &= \frac{2coverage(S,G)}{k} (\frac{balance}{c(S)} - \lfloor\frac{balance}{c(S)} \rfloor) c(S)\\
	                    &= \frac{2(k-1)}{k} (\frac{balance}{c(S)} - \lfloor\frac{balance}{c(S)} \rfloor) c(S)\\
	                    &\leq \frac{2(n-1)}{n} (\frac{balance}{c(S)} - \lfloor\frac{balance}{c(S)} \rfloor) c(S)\\
	                    &= \frac{2(n-1)}{n} R.
\end{align*}

Therefore, the total cost used by the approximation algorithm to reach MST target weight of $T$ is bounded by
\begin{align*}
B_T - R +	\sum_{i=1}^qc(P_i) &\leq B_T - R +\frac{2(n-1)}{n} R\\
&= B_T + (1-\frac{2}{n}) R\\
&\leq 2(1-\frac{1}{n})B_T.
\end{align*}

The theorem follows immediately.

\end{proof}

\subsection{An $\frac{opt}{2}-1$-Solution for Budgeted MMST}

Recall that in an $\frac{opt}{2}-1$-solution, the increase in the weight of the MST is at least $\frac{opt}{2}-1$, where $opt$ is the optimal increase in the weight of the MST for the given budget $B$.

\begin{theorem}
 \label{GreedyAlgo} 
	Budgeted MMST has a polynomial-time algorithm that produces an $\frac{opt}{2}-1$ solution, where $opt$ is the optimal increase in weight of the MST. 
\end{theorem}

\begin{proof}
Again, we follow the algorithm in \cite{FS96}. Let $S$ denote the last set of edges computed by this algorithm  and \textit{balance} denote the remaining budget. The set $S$ induces a partition $P_1,\ldots, P_k$ of some graph $\tilde{G}_{w_l}$. In the last iteration, the algorithm in \cite{FS96} lifts all the edges in $S$ by $\frac{balance}{c(S)}$. If $\frac{balance}{c(S)} \in \mathbb{N}$, then we are done. Otherwise, we lift all the edges in $S$ by $\lfloor\frac{balance}{c(S)} \rfloor$.

\noindent After this lifting operation, the remaining budget $R$ is defined in (\ref{eq:budget}) and the weight of the minimum spanning trees increases by

$$A+coverage(S,G) \lfloor\frac{balance}{c(S)} \rfloor,$$

\noindent where $A$ denotes the total increase in the MST weight before lifting $S$.  Therefore, the remaining increase $I$ of the weight of the minimum spanning trees in comparison with the algorithm in \cite{FS96} is

$$I = coverage(S,G) (\frac{balance}{c(S)} - \lfloor\frac{balance}{c(S)} \rfloor).$$
 
Since the algorithm in \cite{FS96} optimally uses the budget fractionally, the optimal increase in the MST using continuous expenditure of the budget $B$ is 

$$A + coverage(S,G) \lfloor\frac{balance}{c(S)} \rfloor + I,$$

\noindent which is an upper bound on $opt$, the optimal increase in the weight of the MST for the given budget $B$.

Assume that the shores $P_i$ are ranked by increasing costs $c(\delta(P_i))$.  Suppose first that there exists an index $q$ such that $c(\delta(P_1\cup\cdots\cup P_q)) \le R < c(\delta(P_1\cup\cdots\cup P_{q+1}))$. 
%The case where  $c(\delta(P_1))> R$ can be handled similarly.  
Lift all the edges in $\delta(P_1),\ldots,\delta(P_q)$ by one. Since \textit{tolerance}$(S,G) \ge \frac{balance}{c(S)} > \lfloor\frac{balance}{c(S)} \rfloor$, it follows that the weight of the minimum spanning trees increases by at least $q$ in this step. 

By the ordering of the node sets  $P_i$, we have

$$\frac{1}{q+1}\sum_{i=1}^{q+1}c(\delta(P_i)) \leq \frac{1}{k}\sum_{i=1}^kc(\delta(P_i)) = \frac{2c(S)}{k}.$$

\noindent Therefore,

\begin{align}\label{eq:lb}
	q+1 &\geq \frac{k\sum_{i=1}^{q+1}c(\delta(P_i))}{2c(S)} \nonumber\\
	&\geq \frac{k c(\delta(P_1\cup\cdots\cup P_{q+1}))}{2c(S)} \nonumber\\
		&> \frac{k R}{2c(S)} \nonumber\\
		&= \frac{k}{2}(\frac{balance}{c(S)} - \lfloor\frac{balance}{c(S)} \rfloor) \nonumber\\
		&\geq \frac{coverage(S,G)}{2}(\frac{balance}{c(S)} - \lfloor\frac{balance}{c(S)} \rfloor) \nonumber\\
		&= \frac{I}{2}.
\end{align}

\noindent This implies $q >  \frac{I}{2} -1$, and hence the approximation algorithm increases the weight of the minimum spanning trees increases by at least 

\begin{align*}
A+coverage(S,G) \lfloor\frac{balance}{c(S)} \rfloor + q &> \frac{1}{2}(A+coverage(S,G) \lfloor\frac{balance}{c(S)} \rfloor +I) -1\\
&\geq \frac{opt}{2}-1.
\end{align*}

In the case where  $c(\delta(P_1))> R$, since $c(\delta(P_1)) \leq \frac{1}{k}\sum_{i=1}^kc(\delta(P_i))$, we have

$$\frac{kR}{2c(S)} < \frac{kc(\delta(P_1))}{2c(S)} \leq 1.$$

The inequalities in (\ref{eq:lb}) imply that $\frac{I}{2} \leq \frac{kR}{2c(S)}$ and thus $\frac{I}{2} \leq 1$. Therefore, by doing nothing we are still within an additive one of the optimal value of $I$.

\begin{align*}
A +coverage(S,G) \lfloor\frac{balance}{c(S)} \rfloor &\geq A+ coverage(S,G) \lfloor\frac{balance}{c(S)} \rfloor+ \frac{I}{2}-1\\
& \geq \frac{1}{2}(A+coverage(S,G) \lfloor\frac{balance}{c(S)} \rfloor +I) -1\\
&\geq \frac{opt}{2}-1.
\end{align*}

\iffalse
In the case where  $c(\delta(P_1))> R$, we have $kR < 2c(S)$ but $I = \frac{(k-1)R}{c(S)} < 2$. Thus by doing nothing we are still within an additive one of the optimal value of $I$.
\fi
\end{proof}

\section{Hardness of MMST}
\label{sec:hardness}

If we impose upper bounds on the extent to which edge weights can be increased, we can encode a $k$-cut problem using MMSTU (the version of MMST with such upper bounds).
\begin{lemma}
The Budgeted MMSTU problem is NP-complete.
\end{lemma}
\begin{proof}
The problem is clearly in NP. We now reduce min $k$-cut to this problem. Let $(G,k)$ be an instance of $k$-cut. For every edge $e\in E(G)$, consider assigning a cost  $c(e) \equiv 1$, a weight $w(e) \equiv 0$ and an upperbound $u(e) \equiv 1$. For $B\in [0, n]$, consider solving the instance of MMSTU with budget $B$ and let $E_B$ be the set of edges giving by the solution that raises the MST to weight $T_B$. Since $G$ starts with a MST of weight $0$ and every edge in $S_B$ has a final weight of $1$, it follows that $coverage(S_B, G)=T_B$. Then,, it follows that the smallest $B$ such that $T_B\ge k-1$ provides a minimum $k$-cut for $G$. 

%Then a solution to an instance MMSTU with  $(G,w,c,u,B)$ has an associated decision problem of ``Can we increase the weight of the MSTs by $k$ within a budget $B$?". This corresponds to finding a $k$-cut within budget $B$, since a scheme $x$ upgrades the MST by $k$ iff $\{e | x(e) = 1\}$ contains a $k$-cut. 
\end{proof}

%%%%%%%%%%%%%%%%%%%%%%%%%%%%%%

What is perhaps a bit surprising is that even without upper bounds and all upgrade costs being unit, MMST is NP-hard.
\begin{theorem}
\label{th:MMSTNP}
Budgeted MMST is NP-complete even with unit downgrading costs. 
\end{theorem}

\begin{proof}
    Given an instance of Min $k$-Cut, $(G, k)$, consider the following auxiliary graph $G'$. For every edge $e=uv\in E(G)$, we add a clique of size $n^2$ and add edges from the vertices of the clique  to the vertices $u$ and $v$. Original edges start with weight $0$ and all newly added edges have weight $1$. The upgrading cost $c$ for all edges is $1$. Note that an initial MST needs to connect the newly added $n^2|E(G)|$ vertices and thus has weight at least $n^2|E|$. An MST of weight $n^2|E(G)|$ exists by simply taking a spanning tree of the original zero-weight edges and greedily appending the newly added vertices to this spanning tree.
    
    Suppose $F\subseteq E(G)$ is an optimal min $k$-cut and $|F|=b$. Consider spending $b$ budget to upgrade the edges in $F$ to weight $1$. Since $G\backslash F$ has $k$ connected components in $G$, all the edges with final weight of $0$, $E(G)\setminus F$, also forms $k$ connected components in $G'$. Then it follows that has final weight of an MST in $G'$ is at least $2n|E|+k-1$. One can also easily construct an MST with such weight by first taking a maximal forest using the edges in $E(G)\setminus F$, then add $k-1$ edges in $F$ to form a tree that spas all the original vertices and lastly greedily append the newly added vertices in the clique. Therefore under budget $b$, one can increase the MST value by $k-1$. Then, our goal is to find the least amount of budget $b$ that improves the MST value by $k-1$ and show such solution can be translated to a $k$-cut.
    
    Consider greedily trying all values of $b$ where $0\le b\le |E(G)|< n$ and solve the instances of MMST on $G'$ with budget $b$. Let $b'$ be the smallest $b$ such that the solution to MMST increased the MST value by $k-1$. Such $b'$ exists since raising the weight of every original edge by $1$ increases the MST by at least $ n\ge k-1$. |Let $F'$ be the set of edges whose weight was changed by such an optimal solution.
    
    First, we claim that no newly added edge is in $F'$. Suppose for the sake of contradiction that there exists such edge $e'\in F'$ where $e'$ is also in $E(G')\setminus E(G)$. Let $T$ be an MST after upgrading the edges in $F'$ to their final weight. Since $e'$ is a newly-added edge, let $e\in E(G)$ be its associated original edge. Note that the endpoints of $e$ along with the newly-added clique forms a large clique $K'$ of size $n^2+2$. Since $|F'|\le b'<n^2$, $K'\backslash F'$ is connected. Then, one can construct an MST $T'$ from $T$ that avoids any edges in $E(K)\cap F'$. Consider the path $P'$ in $T'$ that connects the endpoints of $e'$. Note that $P'$ lies in $K'$ and the final weight of every edge in $P'$ is $1$. Then, if we did not upgrade $e'$, the weight of every edge in $P'$ is still at most the weight of $e'$. Then, $T'$ is also an MST if we do not upgrade $e'$, contradicting the minimum choice of $b'$.

    %note that since $|F'|\le b'<n^2$, an MST in with the final weight can always use a spanning tree involving no edges in $F'$ to attach the newly added vertices inside the cliques of size $n^2$ to the original vertices $u, v$. It follows that downgrading any newly added edge does not affect the value of the MST. $F'$ does not contain any newly added edges, otherwise we can reduce the budget and still obtain a $k-1$ increase in the value of MST.
    
    Next, we claim that no original edge is upgraded more than once. The argument is similar to the one before. Suppose for the sake of contradiction that an edge $e\in E(G)$ is upgraded more than once. Let $K'$ be the clique of size $n^2+2$ containing $e$. By similar argument, one can show that there exists an MST $T'$ in $G'$ after upgrading such that $e\notin T'$. Then, one can similarly argue that $T'$ remains an MST if $e$ is only upgraded to $1$, contradicting the minimum choice of $b'$. 
    
    %Once an original edge $uv$ is changed to a weight of $1$, downgrading it once more also does not increase the MST value since the MST can reroute through the added clique and maintain its final weight.
    
    Then, it follows that $|F'|=b'$. Let $T'$ be an MST of $G'$ before upgrading the edges in $F'$. Note that $|E(T')\cap F'|\ge k-1$ since the weight of the final MST is increased by at least $k-1$. We claim that $T=T'\backslash (V(G')\setminus V(G))$ is a spanning tree of $G$. Suppose for the sake of contradiction that $T$ is not connected and contains a cut $\delta(V_1)$. Since $G$ is connected, there exists $e\in \delta(V_1)\cap E(G)$. Let $e'\in \delta(V_1)\cap E(T')$. Note that $e'$ has weight $1$ while $e$ has weight $0$. Then substituting $e'$ with $e$ in $T'$ creates a better MST, a contradiction. 
    
    This implies any MST $T'$ in $G'$ restricted to $G$ is a spanning tree in $G$. Then, $coverage(F', G)\ge k-1$ and thus $F'$ is a $k$-cut in $G$. By our choice of $b'$, it follows that $F'$ is an optimal $k$-cut. 
\end{proof}

\begin{corollary}
If there exists an $\alpha$-approximation of Budgeted Discrete MMST then there exists an $\alpha$-approximation of Max Components. \\
If there exists an $\alpha$-approximation of Targeted Discrete MMST then there exists an $\alpha$-approximation of $k$-cut.
\end{corollary}

\begin{proof}
It follows from the proof of Theorem \ref{th:MMSTNP} that given the auxiliary graph $G'$, any feasible solution that spends budget $0\le b\le n$ to upgrade a set $F'$ and increases the MST by $T_b$ can be transformed into a feasible solution that spends budget at most $b$, only upgrades original edges by at most once, does not upgrade any newly added edges and increases the MST by at least $T_b$. Then, it follows that an $\alpha$-approximation to MMST also produces an $\alpha$-approximation to the Max Component problem. Similarly, an $\alpha$-approximation to the Targeted Discrete MMST also provides an $\alpha$ approximation to the Minimum $k$-Cut Problem. 

\end{proof}
%What happens when all edges start with the same weight?
\paragraph{Hardness of Approximation.}
Let $G$ be a $d$-regular undirected unweighted graph. The edge expansion $\phi(S)$ of $S \subset V(G)$ is defined as :
$$\phi(S)= \frac{|E(S,V(G) \setminus S)|}{d \min\{|S|,|V \setminus S|\}}$$
where $E(S, V(G) \setminus S)$ is the set of edges across the partition $(S, V(G) \setminus S)$. 

\begin{problem}[Small Set Expansion (SSE)]

Given a regular graph $G $, let $\delta, \eta\in (0, 1)$. The $SSE(\delta, \eta)$ problem is to distinguish between: 

\begin{enumerate}
    \item \textbf{(Completeness)} There exists $S \subset V$ of size $\delta|V|$ such that $\phi(S) \le \eta$.
    \item \textbf{(Soundness)} For every $S \subset V$ of size $\delta |V|$, $\phi(S) \ge 1-\eta$.
\end{enumerate}
\end{problem}

The author in \cite{Manu18} showed that assuming SSE is hard given any $\eta$ for some $\delta$, which is called the SSE hypothesis, or SSEH, $k$-cut is also hard to approximate to a $2-\varepsilon$ factor. We then provide a similar result for the Max Component Problem.

\begin{theorem}
Assuming SSEH, it is also hard to approximate Maximum Components to within $(\frac{1}{2} + \varepsilon)$ factor of the optimum for every constant $\varepsilon > 0$.
\end{theorem}

\begin{proof}
We reduce SSE to Max Components. Given an instance of $SSE(\delta,\eta)$ on a $d$-regular graph $G$ with sufficiently large $n=|V(G)|$, let $\varepsilon>0$ be a value dependent on $n, \delta$ and $\eta$. We specify later their relationship. 
We will show that if $B = (\frac{1}{2}+\eta)\delta d n$, then solving Maximum Components with approximation ratio $(\frac{1}{2} + \varepsilon)$ on $G$ with budget $B$ is sufficient to solve this instance of $SSE$. First, we make the following observations about Completeness and Soundness. 

\textbf{\em(Completeness)} \quad If there exists $S\subset V$ of size at most $\delta n$ such that $\phi(S) \leq \eta$, then consider partition the graph into $|S| +1$ groups where the first group is $V \setminus S$ and each of the remaining groups contains a single vertex from $S$. The edges between the partitions are those in $E(S, V \setminus S)$ and the edges within the set $S$. There are $d|S|\phi(S) \leq \eta d |S|$ edges of the former type and only at most $d |S|/2$ of the latter. Hence, the number of edges between this partition is at most $(1/2 + \eta)d|S| \le (1/2 + \eta)d \delta n$. This implies there exists a $k'$-cut with at most $B$ edges where $k' = \delta n +1$.

\textbf{\em(Soundness)} \quad Suppose that for every $S \subset V$ of size $\delta n$, $\phi(S) \geq 1-\eta$. Let $k >\delta n(\frac{1}{2} + \frac{\varepsilon}{2})$ and $T_1,...,T_k \subset V$ be any $k$-partition of the graph. Assume without loss of generality that $|T_1|\leq ... \leq |T_k|$. Let $A = T_1 \cup ... \cup T_i$ where $i$ is the maximum index such that $|T_1 \cup  ... \cup T_i| \leq \delta n$.

We can then add up to $x\delta n$ nodes to $A$, where $(1-x)\delta n$ is the size of $A$, to obtain a set $A'$ of size exactly $\delta n$. Note that, since  $|A'| = \delta n$, $\phi(A')\geq 1 - \eta$, and so $E(A',V\setminus A') \geq (1-\eta)d |A'| =(1-\eta)d \delta n $.
This implies that $|E(A,V\setminus A)| \geq (1-\eta)d \delta n - x\delta n d$. We would like to prove that we always have $(1-\eta)d \delta n - x \delta n d \geq B = (\frac{1}{2}+\eta)\delta d n$. That is, that $x \leq  (1-\eta) - (\frac{1}{2}+\eta) = \frac{1}{2} - 2\eta$. Thus we wish to prove that $|A| \geq (\frac{1}{2}+2\eta)\delta n$.

Suppose for the sake of contradiction that $|A| < (\frac{1}{2}+2\eta)\delta n$.

Since $|A \cup T_{i+1}| > \delta n$, we have $|T_{i+1}| > (\frac{1}{2} - 2 \eta) \delta n$. Furthermore, since $A= T_1 \cup ... \cup T_i$, $i \leq |A|<(\frac{1}{2}+2\eta)\delta n$. 

Then,\begin{equation*}
\begin{split}
n & = |T_1 \cup ... \cup T_k|  \\
& > |T_{i+1} \cup ... \cup T_k | \geq (k-i) \cdot |T_{i+1}| \\
& >  ((\frac{1}{2} + \frac{\epsilon}{2})\delta n -(\frac{1}{2} + 2\eta)\delta n )((\frac{1}{2}-2\eta)\delta n) = (\frac{\epsilon}{2}-2\eta) \delta^2 n^2 (\frac{1}{2}-2\eta)
\end{split}
\end{equation*}

However, if $\eta < \min(\frac{\epsilon}{4},1/4)$, then this is positive and tends towards infinity as $n$ goes to infinity. Thus for sufficiently large $n$, this inequality is false, a contradiction. This proves that indeed, $|A| \geq (\frac{1}{2}+2\eta)\delta n$.

Thus this means that $|E(A,V\setminus A)| \geq B$, which means that any $k$-cut has at least $B$ edges when $k> \delta n(\frac{1}{2} + \frac{\varepsilon}{2})$. \\

In other words, if we are in the completeness case, then there exists a $B$ cost cut with  $\delta n +1$ components, whereas in the soundness case, all $B$ cost cuts give at most $\delta n(\frac{1}{2} + \frac{\varepsilon}{2})$ components. The gap between the two is more than $\frac{1}{2} + \varepsilon$ and so an approximation of Maximum Components to within $(\frac{1}{2} + \varepsilon)$ can distinguish between the two cases. This concludes the proof.
\end{proof}

\begin{corollary}
Assuming SSEH, Budgeted Discrete MMST is NP-hard to approximate to a $\frac{1}{2} + \varepsilon$-factor for $\varepsilon > 0$. \\
Assuming SSEH, Targeted Discrete MMST is NP-hard to approximate to a $2 - \varepsilon$-factor for $\varepsilon > 0$.
\end{corollary}

\section{MST Fortification with all edges starting with the Same Weight}
\label{sec:samestart}
In the previous section, we have shown that the unit MST upgrading problem is hard even if all the costs are unitary. However, it is unknown if the problem remains hard if all initial weights start off the same. Thus, consider the special targeted version of unit MST upgrading where all the initial weights of the graph start with the same weight. Without loss of generality, assume the weights all start with $0$. First we provide a polytime algorithm that solves this problem with a fixed target MST value of $T$. For this section, we drop the term $G$ and use $coverage(F)$ to denote the coverage of set $F$ in $G$ with initial weights of zero.

%*** Add motivation since we don't know complexity of all starting with same weight and unit costs for all edges***

%*** Also works for Budgeted version. Use observation $B \geq T$.***

\begin{theorem}
\label{boundedT}
Let $G$ be a graph whose edges have weight $0$ and upgrading costs of $c_e$. Given a fixed target MST value of $T$, there exists an algorithm that runs in $n^{O(T)}$ time and upgrades edges to produce a final MST value of $T$ while minimizing the total cost of upgrading.
\end{theorem}

Note that since all the costs are integral, if it takes budget $B$ to raise the MST by $T$, then $B\ge T$. Thus, given a bounded constant budget $B$, the special budgeted version of unit MST upgrading can also be solved exactly in polytime.

To prove the theorem, we require a supermodular property of coverage. 
\begin{lemma}{(Supermodularity of coverage)}

\label{supmod}
    For any $F, F'\subseteq E(G)$,  $coverage(F)+coverage(F')\le coverage(F\cup F')+coverage (F\cap F')$. 
\end{lemma}

\begin{proof}\label{supmod}
    First we decompose our sets by edge weight $F = \cup F_{w_i}$ and $F = \cup F'_{w_i}$ where $F_{w_i}$ and $F'_{w_i}$ are the set of edges with weight $w_i$ that are also in $F$ and $F'$ respectively. Then 
    $$coverage(F)+coverage(F') = coverage(\cup F_{w_i})+coverage(\cup F'_{w_i})$$
    $$= \sum_i (coverage(F_{w_i})+coverage(F'_{w_i}))$$
    and
    $$coverage(F\cup F')+coverage (F\cap F') = coverage(\cup_i (F_{w_i} \cup F'_{w_i})) + coverage(\cup_{i,j} (F_{w_i} \cap F'_{w_j})) $$
       $$=  coverage(\cup_i (F_{w_i} \cup F'_{w_i})) + coverage(\cup_{i} (F_{w_i} \cap F'_{w_i})) $$
       $$= \sum_i (coverage(F_{w_i} \cup F'_{w_i}) + coverage(F_{w_i} \cap F'_{w_i})) $$
    Note that the second equality follows from the fact that $F_{w_i} \cap F'_{w_j} = \emptyset$ for $i \neq j$.
    It now suffices to show that for $F$ and $F'$ of a single weight class, the inequality holds. 
    
    In consequence, we assume without loss of generality that $F$ and $F'$ both contain only edges of weight $w_i$.
    The coverage of any edge-set is the same as the number of additional components created from its deletion in $\widetilde{G_{w_i}}$. Let $\{G_i\}_{i=1}^k$ be the connected components of $\widetilde{G_{w_i}}\backslash (F\cap F')$. Let $F_i, F_i'$ be the set of edges in $G_i$ that belongs in $F, F'$ respectively. 
    Let $a_i, a_i', b_i$ represent the number of additional components created from $G_i$ by deleting $F_i, F_i', F_i\cup F_i'$ respectively. Note that $a_i, a_i', b_i$ also represent the coverage of $F_i, F_i', F_i\cup F_i'$ respectively in $G_i$. 
    Since $F_i\cap F_i'=\emptyset$ by construction, it follows that any spanning tree of $G_i$ must contain $a_i, a_i'$ edges from $F_i, F_i'$ respectively.
    Thus $b_i\ge a_i+a_i'$. In $\widetilde{G_{w_i}}$, the coverage of $F$ can be viewed as first deleting $F\cap F'$ and then deleting the $F_i$'s in sequence and counting how many additional components it creates. 
    In other words $coverage(F)=k-1+\sum_{i = 1}^k a_i$. We can obtain similar equations for $coverage(F'), coverage(F\cup F')$. 
    
    Then, 
    \begin{equation*}
    \begin{split}
    coverage(F)&+coverage(F') =(k-1+\sum_ia_i)+(k-1+\sum_i a_i') \\
    & \le (k-1)+(k-1+\sum b_i) =coverage(F\cap F')+coverage(F\cup F')
    \end{split}
    \end{equation*}
\end{proof}

To prove Theorem \ref{boundedT}, we also need to solve the following variant of the knapsack problem:

\begin{problem} [Unbounded Knapsack Problem]
For $i\in [n]$, let $w_i$ and $p_i$ be respectively the weight and the profit of item $i$. Given fixed capacity $W$, find $x_i\in \mathbb{Z}_{\ge 0}$ that maximizes $\sum_{i=1}^n p_ix_i$ subject to $\sum_{i=1}^n w_ix_i\le W$. 
\end{problem}

\begin{lemma}
\label{unboundknap}
Given an instance of the Unbounded Knapsack Problem, if each profit $p_i$ is of order $O(n)$, then a solution can be found in time $O(n^2p^2)$ where $p=\max_i p_i$. 
\end{lemma}

A solution for this above lemma is proposed in \cite{stackex}; we formally reproduce the proof below for completeness.

\begin{proof}
    Given an instance the Unbounded Knapsack Problem, let $x^*$ be an optimal solution and $P^*$ be the optimal profit. Assume without loss of generality that $p_1/w_1 = \max_{i=1}^n \{p_i/w_i\}$. In other words, item $1$ gives the most bang-for-buck.  Let $P=\sum_{i=2}^n p_1p_i$. We break into two cases depending on how $P^*$ compares to $P$. 
    
    \paragraph{Case 1} Suppose $P^*\le P$. Then, consider the following dynamic programming. For $i\in [n], 1\le p\le P$, let $f(i, p)$ be the least possible weight of a solution such that it contains at least one copy of item $i$ and its profit is exactly $p$. For $i\in [n]$, let $f(i, p_i)=w_i$ and let $f(i, p) = \infty$ for all $p< p_i$. This corresponds to the fact that if a solution contains item $i$, its profit must be at least $p_i$. For $i\in[n], p_i<p\le P$, let $f(i, p) = \min_{j\in [n]} \{f(j, p-p_i)\} + w_i$. Intuitively, this recursion says that if a solution has profit $p$ and contains item $i$, then removing it results in a solution with profit $p-p_i$ and weight $f(i, p)-w_i$. Thus, searching for the best way to achieve profit $p-p_i$ also leads to a solution for $f(i, p)$. Since $i, p$ are polynomially bounded, all values of $f$ can be computed in polytime. For $1\le p\le P$, let $w(p)=\min_{i\in [n]} \{f(i, p)\}$. Then, simply find the largest value of $p$ such that $w(p)\le W$.
    
    \paragraph{Case 2} Suppose $P^*>P$. Without loss of generality, assume $x^*$ is an optimal solution that maximizes the value $x^*_1$. First we claim that $x^*_1\neq 0$. Since $\sum_{i=1}^n x_i^*p_i = P^* > P = \sum_{i=1}^n p_1p_i$, there exists $i\in [n]$ such that $x_i^*\ge p_1$. If $i=1$, then we are done. Otherwise, since $p_1/w_1\ge p_i/w_i$, it follows that by swapping out $p_1$ copies of item $i$ and replace it with $p_i$ copies of item $1$, the profit does not change and the weight does not increase. Since $x^*$ is an optimal solution that maximizes $x^*_1$, it follows that $x^*_1>0$.
    
    Then, it follows that removing a single copy of item $1$ from $x^*$ is an optimal solution for profit $P^*-p_1$. Therefore one can recursively remove $p_1$ from $P^*$ until the profit falls below $P$ and use the solutions from Case 1 to build $x^*$.
    
    In particular, consider $W'$ such that $W-w_1 <W' \le W$. Let $w'$ be the largest integer such that $W'\le w(P)$ and $W'=W-k'w_1$ for some integer $k'$. Note that $x^*$ has total weight $W'$ if and only if there exists $p'$ such that $w(p')=w'$. Then, let $\mathcal{W}$ be the set of weights $W'$ where such $p'$ exists for the associated $w'$. Since $x^*$ contains a copy of item $1$, its total weight must be one of the values in $\mathcal{W}$. For each $W'\in \mathcal{W}$, its associated solution must be $x'$ and an additional $k'$ copies of item $1$. Therefore, its associated profit is $p'+k'p_1$. Then, taking the solution with the largest profit $p;+k'p_1$ is an optimal solution. 
\end{proof}

\begin{proof}[Proof of Theorem \ref{boundedT}]
    Let $w^*$ be an optimal solution representing the final weights of every edge. Let $b^*$ be the upgrading cost to reach the final weights $w^*$. Let $E_i^*$ represent the set of edges whose final weight is at least $i$. Then, we can decompose $w^*$ into a sequence of lifts where we first lift $E_1^*$ by one unit, then $E_2^*$ and so on. Note that at the time of lifting $E_i^*$, the set $E_i^*$ has the highest weight among all edges of the graph. Then, it follows from Lemma \ref{lem:sumeq}, lifting $E_i^*$ at time $i$ increases the MST by exactly the same amount as the coverage of $E_i^*$ at the beginning where all weights are the same. Thus, it motivates us to only look at the coverage of sets at the initial stage where all weights are $0$. 
    
    Define $F_i$ to be a set of edges in $G$ with the least upgrading cost such that $G\backslash F_i$ creates $i$ additional components for $i\le T$. Note that finding these sets is equivalent to solving $T$ iterations of the Min $i$-Cut problem, which can be solved in $n^{O(T)}$ time. Let $c(F_i)$ represent the cost to upgrade the set $F_i$ by one unit. It follows that an optimal solution is a combination of these cuts that maximizes the sum of the coverage while ensuring the cost is within budget. 
    
    Then, consider the following instance of the Unbounded Knapsack Problem. Consider $T$ items $F_1, ..., F_T$ where $F_i$ has weight $c(F_i)$ and profit $i$. It follows from Lemma \ref{unboundknap} that a solution can be obtained in polynomial time. 
    
%    First, we construct a solution $s^*$ from the sets $E_i^*$ that is feasible to the above knapsack problem. Let $x(F_j)=n_j$ where $n_i$ is the number of sets $E^*_i$ such that $coverage(E^*_i)=j$. Due to our choice of $F_j$, it follows that $x^*$ has a final profit of $T$ and a cost of at most $b^*$. Then, it suffices to solve the knapsack problem, obtain a solution $x$ and show that its solution can be converted to a sequence of lifts that raises the MST by as much as $T$ while costing at most the weight of $x$.

%    Since $T$ is bounded, this knapsack problem can be solved using dynamic programming. For each potential value $t\le T$ and item type $i\le T$, let $c(t, i)$ represent the least weight to reach exactly profit $t$ while using at least one copy of $i$. Then, initiate with $c(j, i)=\infty$ for all $j< i$, $c(i, i)=c(F_i)$ and recursively define $c(t, i)=\min_j\{c(t-i, j)\}+c(F_i)$. Lastly, by taking the least weight to reach profit $T$, one can easily obtain a solution to the knapsack problem. 
    
    Let $x^*$ be an optimal solution to the knapsack problem. Now, we slightly modify the knapsack solution so that the subsets form a chain. Suppose there exists $i, j$ where neither $F_i, F_j$ are subsets of the other one and their $x^*$ values are non-zero. Then, by Lemma \ref{supmod}, we can use $F_i\cap F_j, F_i\cup F_j$ instead since it does not change the total upgrading cost. Then, we can perform these uncrossing operations to obtain a family of edge-sets $F_k^*$ and integers $y_k^*$ such that $\sum_k y_k^* coverage(F_k^*) \ge \sum_i x_i^* coverage (F_i)$. It is worth mentioning that since there are at most $T$ edge-sets with non-zero $x^*, y^*$ values, it takes at most $T^2$ many applications of Lemma \ref{supmod} to obtain the sets $F_k^*$. Since $F_k^*$ forms a chain, we can lift them in sequence from the largest to the smallest by $y_k^*$ amount at a time. It follows from Lemma \ref{lem:sumeq} that every lift increases the MST by exactly $coverage(F_k^*)$. Thus, it follows that this process reaches target $T$ and minimizes the downgrading cost. 
\end{proof}

\begin{corollary}
For every constant $\varepsilon>0$, there exists a $(1/2-\varepsilon)$-approximation algorithm for MMST with uniform starting weights.
\end{corollary}

\begin{proof}
    Let us fix $\varepsilon>0$.
    We now describe an algorithm that runs in polynomial time and that gives a $(1/2-\varepsilon)$-approximation algorithm for MMST with uniform starting weights.
    Let $(G,0,c,B)$ be an instance of MMST with 0 weight for all edges. Let $opt$ be the optimal increase for the MST in an optimal solution.
    We first run the greedy algorithm in Theorem \ref{GreedyAlgo} in polynomial time, giving us a $t$ increase in the MST weight, with $opt \ge t \ge \frac{opt}{2} - 1$.
    If $t \ge \frac{1}{\varepsilon}$, then $opt \ge \frac{1}{\varepsilon}$ and so $t \ge opt(\frac{1}{2} - \frac{1}{opt}) \ge opt(\frac{1}{2} - \varepsilon)$. If $t < \frac{1}{\varepsilon}$, then $\frac{1}{\varepsilon} > \frac{opt}{2} - 1$, and $opt<2+\frac{2}{\varepsilon}$ and so since $\varepsilon$ is fixed, by Theorem \ref{boundedT}, an optimal solution can be obtained in this case in time $n^{O(2+2/\varepsilon)}$ which is polynomial in $n$ for any fixed constant $\varepsilon$.
    
    Thus, in all cases, the algorithm runs in polynomial time and finds a $(1/2-\varepsilon)$-approximation.
\end{proof}

\section{Extension and Open Problem}

We expect our modifications to work in a straightforward manner for the extension of the main problems from MSTs (graphical matroid bases) to general matroid bases following the framework in~\cite{FS98}.
The main open problem from our work is to extend it to the case of directed graphs.

\newpage
%\appendix

\appendix

%\section{Hardness of Bulk Upgrading of Minimum $st$-cuts}
%\label{apdx:bulk}

\end{document}